\documentclass[10pt]{IEEEtran}
\usepackage{epsfig,color,amsmath,cite}
\usepackage{amsthm}
\usepackage{wrapfig}
\usepackage{bm}
\usepackage{epstopdf}
\usepackage{amssymb}
\usepackage{url}
\usepackage{enumitem}
\usepackage{multirow}
\usepackage{booktabs}
\usepackage{tikz}
\usepackage{pgfplots}
\usepackage[linesnumbered,lined,boxed,commentsnumbered,ruled,longend]{algorithm2e}


\setlist[itemize]{leftmargin=*}

\makeatother

\newtheorem{theorem}{Theorem}
\newtheorem{mydef}{Definition}
\newtheorem{mylem}{Lemma}

\newtheorem{myasm}{Assumption}

\newcommand{\mat}[1]{\boldsymbol{#1}}

\newcommand{\bmat}[1]{\begin{bmatrix} #1 \end{bmatrix}}

\providecommand{\mA}{\ensuremath{\mat{A}}}
\providecommand{\mB}{\ensuremath{\mat{B}}}
\providecommand{\mC}{\ensuremath{\mat{C}}}
\providecommand{\mD}{\ensuremath{\mat{D}}}

\providecommand{\mI}{\ensuremath{\mat{I}}}

\providecommand{\mL}{\ensuremath{\mat{L}}}

\providecommand{\mP}{\ensuremath{\mat{P}}}

\providecommand{\mY}{\ensuremath{\mat{Y}}}

\newcommand{\m}{\boldsymbol}
\allowdisplaybreaks[4]
\usepackage[colorlinks = true,
linkcolor = blue,
urlcolor  = blue,
citecolor = blue,
anchorcolor = blue]{hyperref}

\newcommand{\mc}[1]{\mathcal{#1}}

\usepackage{graphicx}
\usepackage{color}
\usepackage{amsfonts}
\usepackage{float}
\usepackage{amsmath}
\usepackage{amsthm}
\usepackage{epstopdf}
\usepackage{changepage}
\usepackage{latexsym}
\usepackage{tikz}
\usepackage{cases}
\usepackage{longtable}
\usepackage{supertabular}
\usepackage{tabu}
\usepackage{multirow}
\usepackage{multicol}
\usepackage{booktabs}
\usepackage{arydshln}
\usepackage{setspace}

\usepackage[noabbrev]{cleveref}

\usepackage{mathtools}

\DeclarePairedDelimiter\abs{\lvert}{\rvert}%
\DeclarePairedDelimiter\norm{\lVert}{\rVert}%

\makeatletter
\let\oldabs\abs
\def\abs{\@ifstar{\oldabs}{\oldabs*}}
\let\oldnorm\norm
\def\norm{\@ifstar{\oldnorm}{\oldnorm*}}
\makeatother

\DeclareMathAlphabet\mathbfcal{OMS}{cmsy}{b}{n}

\usepackage{stackengine}
\newcommand\barbelow[1]{\stackunder[1.2pt]{$#1$}{\rule{.8ex}{.075ex}}}

\newcommand*{\permcomb}[4][0mu]{{{}^{#3}\mkern#1#2_{#4}}}

\newcommand*{\comb}[1][-1mu]{\permcomb[#1]{C}}

\usepackage{graphicx}
\usepackage{float}
\usepackage[caption = false]{subfig}

\usepackage[labelfont=bf]{caption}

	\title{\vspace{0.4cm} \LARGE  \textbf{Characterizing the Nonlinearity of Power System Generator Models}}
	\author{Sebastian A. Nugroh$\text{o}^*$, Ahmad F. Tah$\text{a}^*$, and Junjian Q$\text{i}^{\dagger}$\thanks{
			$^*$Department of Electrical and Computer Engineering, The University of Texas at San Antonio, TX 78249. Emails: Sebastian.Nugroho@my.utsa.edu, Ahmad.Taha@utsa.edu. $^{\dagger}$Department of Electrical and Computer Engineering, University of Central Florida, Orlando, FL 32816 USA. Email: Junjian.Qi@ucf.edu. }}
\begin{document}
	\maketitle

	\begin{abstract}
	Power system dynamics are naturally nonlinear. The nonlinearity stems from power flows, generator dynamics, and electromagnetic transients. Characterizing the nonlinearity of the dynamical power system model is useful for designing superior estimation and control methods, providing better situational awareness and system stability. In this paper, we consider the synchronous generator model with a phasor measurement unit (PMU) that is installed at the terminal bus of the generator. The corresponding nonlinear process-measurement model is shown to be locally Lipschitz, i.e., the dynamics are limited in how fast they can evolve in an arbitrary compact region of the state-space. We then investigate different methods to compute Lipschitz constants for this model, which is vital for performing dynamic state estimation (DSE) or state-feedback control using Lyapunov theory. In particular, we compare a derived analytical bound with numerical methods based on low discrepancy sampling algorithms. Applications of the computed bounds to dynamic state estimation are showcased. The paper is concluded with numerical tests.
	\end{abstract}


	\section{Introduction}
Single- and multi-machine power system models have been thoroughly developed and explored in the literature of power systems \cite{kundur1994power,sauer2017power}. These models describe the electromagnetic transients of interconnected generators in transmission networks, ranging from the simple second-order swing equations to tenth or higher-order, nonlinear differential algebraic equation representations.  
%	For example, the textbook by Sauer, Pai, and Chow~\cite{sauer2017power} provides an excellent, thorough exposition on the derivation of these models. % (for single and multi-machine grid models). %The authors also include a discussion on how to derive---or transition to---simpler, lower order models from higher order ones. 
	%A \textit{generator-only model}\footnote{A generator-only model is one that eliminates the algebraic power flow equations and load dependence, thereby transforming a differential algebraic model of interconnected power networks into a differential equation one. This approach is common in various feedback control and dynamic state estimation studies of power systems.} 
%	After eliminating the non-generator buses and considering 
	By considering that phasor measurement units (PMUs) are installed at the terminal buses of selected generators, a nonlinear, dynamical power system model can be generally described as follows %by the following set of equations:
	%of power systems can be written as 
	\begin{align}{\label{equ:model-ps}}
		\dot{\m x}(t)=\tilde{\m f}(\m x,\m u),\quad\;\;
		\m y(t) = \tilde{\m h}(\m x, \m u), 
	\end{align} 
	where $\m x(t)$ represents the dynamic state vector, $\m u(t)$ depicts the known or unknown input vector, and $\m y(t)$ models the output measurements from PMUs~\cite{qi2018comparing}. %locations with 
	%phasor measurement unit (PMUs).% installations. 
	
%	\cite{ekf1,ekf2,zhao1,valverde2011unscented,pwukf1,singh2014decentralized,zhou2013,qi2018dynamic,taha2018risk}
	
	%Transient stability 
	Dynamic modeling of power systems is important because it can guide the development of open- or closed-loop control algorithms. 
	This is in addition to dynamic state estimation (DSE) routines \cite{ghahremani2011dynamic,zhao1,zhou2013estimation,qi2018dynamic,taha2018risk}---under the presence %or absence 
	of unknown inputs, disturbances, faults, and noise. 
%	in the process dynamics as well as sensor faults, noise, and bad data in PMU outputs. 
	The majority of state-feedback control algorithms that utilize low- or high-order linearized power system models, or low-order nonlinear models, such as proportional-integral control, linear quadratic regulator, $\mc{H}_2$, $\mc{H}_{\infty}$, mixed $\mc{H}_2$/$\mc{H}_{\infty}$, and model predictive control have been applied to power systems. %models. 
	Our recent work on robust control in power systems succinctly lists the main control algorithms used in the above context~\cite{taha2018robust}. 
	
	Although advanced static state estimation technique for power systems are still being developed, for example, to minimize the impact of cyber attacks \cite{Khatibi2018}, DSE is considered to be superior due to its capability for performing state estimation in almost real time. Particularly for DSE and state observers in power systems, many studies have used %significantly more complex 
	power system models with different levels of details, with overwhelming focus on Kalman filters and its different variants; 
	%---and its derivatives---as well as signal processing-based techniques; 
	see our recent paper~\cite{qi2018comparing} and the references therein for a comparison of DSE approaches in power systems. Surprisingly, systems-theoretic observer designs are less common in the literature of power system DSE---especially when compared with feedback control algorithms. 
	%Specifically, we illustrate in~\cite{qi2018comparing} that an observer for nonlinear power system models can perform exceptionally well when compared to classic DSE techniques. 
	
	With that in mind, we are interested in utilizing observer-based approach to perform DSE while considering the nonlinear dynamics model of power systems. To do so, first we need to classify the nonlinearities in power system models (that is, $\tilde{\m f}(\cdot)$ and $\tilde{\m h}(\cdot)$ in \eqref{equ:model-ps}) as they can be classified into an abundance of function sets such as Lipschitz continuous, one-sided Lipschitz, quadratic inner-boundedness, or bounded Jacobian~\cite{abbaszadeh2010,phanomchoeng2010,qi2018comparing,Jin2018}. Here, we put our interest on Lipschitz nonlinearity because of its simplicity. It is worth noticing that the majority of Lipschitz-based observer for nonlinear systems cannot cope with relatively large (or \textit{conservative}) Lipschitz constants. Because of this reason, in this paper we \textit{(a)} \textit{investigate} different methods (analytical and numerical) to compute/approximate Lipschitz constants for a synchronous generator model, \textit{(b)} compare the Lipschitz constants obtained from the two methods, and \textit{(c)} check their applicability for performing DSE on single generator using a vintage Lipschitz-based observer proposed in\cite{phanomchoeng2010}, which is akin to the Luenberger observer.

The presented research here is motivated by the work of Siljak \textit{et al.}~\cite{siljak2002robust} on robust decentralized control of power systems. By proving that the nonlinearity in the considered model is \textit{quadratically bounded}, decentralized control framework considering nonlinear model of power systems are developed in \cite{siljak2002robust}. The ideas from~\cite{siljak2002robust} are then extended by Lian \textit{et al.}~\cite{lian2017decentralized} with applications to enhancement of damping ratios of the inter-area oscillation through decentralized robust control~\cite{lian2017decentralized}. As for DSE methods, our recent work~\cite{qi2018comparing} assumes that a higher-order, multi-machine power system model is {one-sided Lipschitz} and {quadratically inner-bounded}. This assumption is then followed by designing a DSE method for uncertain power systems. 
%that is empirically shown to perform well. 
The study by Jin \textit{et al.}~\cite{Jin2018} considers the problem of designing a DSE method for a general class of nonlinear Lipschitz dynamic systems, with applications to interconnected power systems using the second-order swing equations and a linear measurement model. %, and only sinusoidal nonlinearities in the network. 
The authors show that the proposed DSE method is less conservative than its counterparts, making it attractive for large-scale systems.

	In short, this paper aims to investigate different methods to obtain the Lipschitz constants which can be used to perform DSE on a single generator using Lipchitz-based observer. The paper contributions and organization are summarized as follows. First, we reproduce the fourth-order generator model with PMU measurements as outputs (Section~\ref{sec:model}). This model has been used in DSE studies and shown to be able to estimate the nonlinear behavior of the generator \cite{zhou2013estimation}. Second, we propose an analytical method to compute Lipschitz constants for the process and PMU measurement models, which depends on the bounds of the state and input vectors (Section \ref{sec:class}). % and reflects the parameters of the generator and the fact that PMUs measure terminal voltages and currents. 
	%Third, and as a departure from the locally Lipschitz function set classification, we derive the bounded Jacobian constants for the generator model (Section~\ref{sec:class}). The motivation behind also considering the bounded Jacobian classification is that Lipschitz constants can be too conservative when utilized for state-feedback control and DSE in dynamic systems as reported in~\cite{zemouche2013}. While this can be true for general dynamic systems, another goal of this work is to investigate whether the derived Lipschitz constant is useful for power system applications, or whether it is too conservative. 
	Third, we propose a simple sampling-based numerical algorithm that, in theory, could generate less conservative Lipschitz (Section~\ref{sec:numerical-algorithms}). Fourth, we briefly present a Lipschitz-based observer that is crucial for performing DSE (Section \ref{sec:observer}). Finally, in addition to comparing the results of obtaining Lipschitz constants using the two aforementioned methods, we present an application of the proposed theoretical/computational bounds to DSE of generator states given PMU measurements (Section~\ref{sec:simulations}).

%	Numerical examples are given showcasing the discrepancy between the theoretical and numerical methods. 
%	Finally, we present an application of the proposed theoretical/computational bounds to DSE of generator states given PMU measurements (Section~\ref{sec:simulations}). 
	
	%The paper's organization is given next.
	
	%Section~\ref{sec:model} reproduces the generator model with a nonlinear PMU measurement model. Section~\ref{sec:class} presents the paper's main contributions---the derivation of the Lipschitz and bounded Jacobian constants of the generator model. Section~\ref{sec:observer-design} proposes novel observer designs considering the derived classifications of the nonlinearities. Section~\ref{sec:numerical-algorithms} explores simple algorithms that numerically compute the Lipschitz constant---as an alternative to the theoretical developments in Section~\ref{sec:class}. Finally, Section~\ref{sec:simulations} presents numerical simulations showcasing the quality of the theoretical and numerical methods.
	
	%The paper's main contributions are: \textit{(a)} showing that the fourth order generator model is locally Lipschitz, \textit{(b)} deriving the bounded Jacobian constants for the same model, \textit{(c)} illustrating applications of power 
	
	%In contrast, this paper's objective is to derive quantities that quantify nonlinearities in power system with potential applications in control/estimation. 
	
	%To assess whether the derived 
	
	%A major uncharted territory 

	\section{Generator Dynamic Model}  \label{sec:model}
	It is usually difficult to directly measure the internal states of a synchronous generator.
%	, such as rotor angle, rotor speed, or transient voltages. 
	By contrast, with PMU installed at the terminal bus of the generator, the voltage and current phasors can be easily measured and can then be used to estimate the internal states of the generator \cite{zhou2013estimation,taha2018risk}. 
%	The estimated internal dynamic states of the generator can provide better situational awareness for the system operators and can also be used for developing local and global control schemes to enhance the system stability \cite{ghahremani2011dynamic}. 
%	For instance, the estimated rotor speed can be used as input signals for excitation systems \cite{singh2016decentralized} or power system stabilizer \cite{zhang2008design} of the synchronous generator. The estimated rotor angle can be used for online angle stability evaluation and fast instability detection \cite{wei2018model}. Besides, dynamic security assessment  \cite{vittal2011next} requires accurate models of the generators and their controllers, which can be provided by advanced DSE approaches that could validate the model and calibrate the parameters \cite{huang2013generator}. %huang2018calibrating	
	Here, we focus on modeling and understanding the nonlinearities of PMU-connected single synchronous generator. 
%This understanding, as discussed above, can then be used to perform a variety of applications. 
%	A model that considers interconnected generator modeling is important as well, but it outside the scope of this work. 
	For generator $i$, the fast sub-transient dynamics and saturation effects are ignored and the generator model is described by the fourth-order differential equations in local $\mathrm{d}$-$\mathrm{q}$ reference frame \cite{sauer2017power} %~\cite{sauer1998power}
	\begingroup
	\allowdisplaybreaks
	\begin{subnumcases}{\label{gen model}}
		\dot{\delta_i}=\omega_i-\omega_0 \\
		\dot{\omega}_i=\frac{\omega_0}{2H_i}\Big(T_{\textrm{m}i}-T_{\textrm{e}i}-\frac{K_{\textrm{D}i}}{\omega_0}(\omega_i-\omega_0)\Big) \\
		\dot{e}'_{\textrm{q}i}=\frac{1}{T'_{\textrm{d0}i}}\Big(E_{\textrm{fd}i}-e'_{\textrm{q}i}-(x_{\textrm{d}i}-x'_{\textrm{d}i})\,i_{\textrm{d}i}\Big) \\
		\dot{e}'_{\textrm{d}i}=\frac{1}{T'_{\textrm{q0}i}}\Big(-e'_{\textrm{d}i}+(x_{\textrm{q}i}-x'_{\textrm{q}i})\,i_{\textrm{q}i}\Big),
	\end{subnumcases}
\endgroup
	%where $E_{\textrm{t}i}=e_{\textrm{R}i}+je_{\textrm{I}i}$ and $I_{\textrm{t}i}=i_{\textrm{R}i}+ji_{\textrm{I}i}$ can be measured by the PMU installed at the terminal bus of generator $i$. 
	%$T_{\textrm{m}i}$, $E_{\textrm{fd}i}$ can be measured or estimated. 
	%For dynamic state estimation of generator $i$, $T_{\textrm{m}i}$, $E_{\textrm{fd}i}$, $I_{\textrm{t}i}$ are used as inputs and $E_{\textrm{t}i}$ is used as the output. 
	where %the physical meaning of the states and their units can be found in~\cite{sauer1998power,kundur1994power}. 
	$\delta_i(t):=\delta_i$ is the rotor angle,
	$\omega_i(t):=\omega$ is the rotor speed in rad/s, and $e'_{\mathrm{q}i}(t):=e'_{\mathrm{q}i}$ and $e'_{\mathrm{d}i}(t):=e'_{\mathrm{d}i}$ are the transient voltage along $\mathrm{q}$ and $\mathrm{d}$ axes; 
	$i_{\mathrm{q}i}(t):=i_{\mathrm{q}i}$ and $i_{\mathrm{d}i}(t):=i_{\mathrm{d}i}$ are stator currents at $\mathrm{q}$ and $\mathrm{d}$ axes;
	$T_{\mathrm{m}i}(t):=T_{\mathrm{m}i}$ is the mechanical torque, $T_{\mathrm{e}i}(t):=T_{\mathrm{e}i}$ is the electric air-gap torque, 
	and $E_{\mathrm{fd}i}(t):=E_{\mathrm{fd}i}$ is the internal field voltage; $\omega_0$ is the rated value of angular frequency, $H_i$ is the inertia constant, and $K_{\mathrm{D}i}$ is the damping factor; 
	$T'_{\mathrm{q0}i}$ and $T'_{\mathrm{d0}i}$ are the open-circuit time constants for $\mathrm{q}$ and $\mathrm{d}$ axes; $x_{\mathrm{q}i}$ and $x_{\mathrm{d}i}$ are the synchronous reactance and $x'_{\mathrm{q}i}$ and $x'_{\mathrm{d}i}$ are the transient reactance respectively at the $\mathrm{q}$ and $\mathrm{d}$ axes. 
	We assume that a PMU is installed at the terminal bus of generator $i$. The mechanical torque $T_{\mathrm{m}i}$ and internal field voltage $E_{\mathrm{fd}i}$ are considered as inputs which can be measured/estimated \cite{zhou2013estimation}. 
	Additionally, we take the current phasor $I_{ti}=i_{\mathrm{R}i} + i_{\mathrm{I}i}$ measured by PMU as inputs which can help decouple generator $i$ from the rest of the network \cite{zhou2013estimation}. The voltage phasor $E_{ti}=e_{\mathrm{R}i} + e_{\mathrm{I}i}$ can also be measured by PMU and is considered as output. 
	The dynamic model (\ref{gen model}) can be rewritten in a general state-space form~(\ref{equ:model-ps})
	%\begin{subnumcases} {\label{n1}}
	%	\dot{\boldsymbol{x}}=\boldsymbol{f}_\mathrm{c}(\boldsymbol{x},\boldsymbol{u}) \\
	%	\boldsymbol{y}=\boldsymbol{h}_\mathrm{c}(\boldsymbol{x},\boldsymbol{u}),
	%\end{subnumcases}
	where the state, input, and output vectors are specified as
	\begin{subequations} \label{xuy}
		\begin{align}
		&\boldsymbol{x} = \bmat{x_1 \quad x_2 \quad x_3 \quad x_4}^\top = \bmat{\delta_i \quad \omega_i \quad e'_{\mathrm{q}i} \quad e'_{\mathrm{d}i}}^\top \\
		&\boldsymbol{u} = \bmat{u_1 \quad u_2 \quad u_3 \quad u_4}^\top = \bmat{T_{\mathrm{m}i} \quad E_{\mathrm{fd}i} \quad i_{\mathrm{R}i} \quad i_{\mathrm{I}i}}^\top \\
		&\boldsymbol{y} = \bmat{y_1 \quad y_2}^\top = \bmat{e_{\mathrm{R}i} \quad e_{\mathrm{I}i}}^\top.
		\end{align}
	\end{subequations}
	\vspace{-0.1cm}
	\allowdisplaybreaks
	\noindent The $i_{\textrm{q}i}$, $i_{\textrm{d}i}$, and $T_{\textrm{e}i}$ in (\ref{gen model}) are functions of $\boldsymbol{x}$ and $\boldsymbol{u}$ as follows
	%{\small
	%	 \begin{subequations}
	\begin{align}
	&i_{\textrm{q}i}=i_{\textrm{I}i}\sin\delta_i+i_{\textrm{R}i}\cos\delta_i =u_4 \sin x_1 + u_3 \cos x_1 \notag \\
	&i_{\textrm{d}i}=i_{\textrm{R}i}\sin\delta_i-i_{\textrm{I}i}\cos\delta_i =u_3 \sin x_1-u_4 \cos x_1\notag \\
	&e_{\textrm{q}i}=e'_{\textrm{q}i}-\frac{S_\textrm{B}}{S_{\textrm{N}i}}x'_{\textrm{d}i}i_{\textrm{d}i} = x_3-\frac{S_\textrm{B}}{S_{\textrm{N}i}}x'_{\textrm{d}i}i_{\textrm{d}i}  \label{temp}\\
	&e_{\textrm{d}i}=e'_{\textrm{d}i}+\frac{S_\textrm{B}}{S_{\textrm{N}i}}x'_{\textrm{q}i}i_{\textrm{q}i} = x_4 + \frac{S_\textrm{B}}{S_{\textrm{N}i}}x'_{\textrm{q}i}i_{\textrm{q}i} \notag \\
	&P_{\textrm{e}i} = e_{\textrm{q}i}i_{\textrm{q}i}+e_{\textrm{d}i}i_{\textrm{d}i} \,\quad T_{\textrm{e}i} = \frac{S_\textrm{B}}{S_{\textrm{N}i}} P_{\textrm{e}i}, \notag
	\end{align}
	%\end{subequations}}
	where  $e_{\mathrm{q}i}$ and $e_{\mathrm{d}i}$ are the terminal voltage at $\mathrm{q}$ and $\mathrm{d}$ axes, and $S_\mathrm{B}$ and $S_{\mathrm{N}i}$ are the system base MVA and the base MVA for generator $i$, respectively.
	The PMU outputs $e_{\textrm{R}i}$ and $e_{\textrm{I}i}$ can be written as functions of $\boldsymbol{x}$ and $\boldsymbol{u}$ as follows
%	\begin{subequations}\label{output equation}
		\begin{align}\label{output equation}
	\hspace{-0.2cm}	&e_{\textrm{R}i}= e_{\textrm{d}i}\sin\delta_i+e_{\textrm{q}i}\cos\delta_i, \;
		e_{\textrm{I}i}=e_{\textrm{q}i}\sin\delta_i-e_{\textrm{d}i}\cos\delta_i.
		\end{align}
%	\end{subequations}
%	Using \eqref{xuy} and \eqref{temp}, we can rewrite \eqref{gen model} and \eqref{output equation} as
%	\begin{subnumcases}{\label{gen model1}}
%		\dot{x}_1=x_2-\alpha_1 \\
%		\dot{x}_2=\alpha_2u_1-\alpha_5x_2-\alpha_3\left(x_3u_4+x_4u_3\right)\sin x_1\nonumber\\
%		\quad\quad \,+\,\alpha_3\left(x_4u_4-x_3u_3\right)\cos x_1 +\alpha_4u_3u_4\cos 2x_1 \,\quad \nonumber \\
%		\quad\quad \,+\,\tfrac{1}{2}\alpha_4\left(u_4^2-u_3^2\right)\sin 2x_1 + \alpha_6 \\
%		\dot{x}_3=\alpha_7 u_2 - \alpha_7 x_3 - \alpha_8 \left(u_3\sin x_1-u_4\cos x_1\right) \\
%		\dot{x}_4=\alpha_{10} \left(u_3\cos x_1+u_4\sin x_1\right) - \alpha_9 x_4, \\
%		%\end{subnumcases}
%		%\vspace{-0.45cm}
%		%\begin{subequations}{\label{output model}}
%		%\begin{align}
%		y_1 = x_3\cos x_1 + x_4\sin x_1 +\beta_1u_3\sin 2x_1 \nonumber \\ 
%		\qquad + \beta_1u_4\cos 2x_1 + \beta_2u_4\\
%		y_2 = x_3\sin x_1 - x_4\cos x_1 -\beta_1u_3\cos 2x_1 \nonumber \\ 
%		\qquad - \beta_1u_4\sin 2x_1 - \beta_2u_3,
%		%\end{align}
%	\end{subnumcases}
%	where the parameters $\alpha_1,\cdots,\alpha_{10}$, $\beta_1$, and $\beta_2$ are given in Appendix~\ref{app:params}.
	%\textcolor{red}{Sebastian: After (7), Please write (6) again as a nonlinear system without the other nonlinear functions, i.e., write it as you write any nonlinear dynamic system $\dot{x}=Ax+Bu+\phi(x,u)$. Do the same thing for the output equation.}
	By substituting \eqref{xuy} and \eqref{temp} to \eqref{gen model} and \eqref{output equation},  
%	and then rearrange the state vector $\m x$ and input vector $\m u$ such that 
	the generator's dynamics can be modeled into the following form
%	and separating the linear terms from the nonlinear ones, the state-space dynamics are equivalent to
	\begin{subnumcases} {\label{eq:state_space_gen}}
		\dot{\boldsymbol{x}}= \mA \m x + \boldsymbol{f}(\boldsymbol{x},\boldsymbol{u}) + \m{B_\mathrm{u}} \m u  \\
		\boldsymbol{y}=\boldsymbol{h}(\boldsymbol{x},\boldsymbol{u}) + \m{D_\mathrm{u}} \m u,
	\end{subnumcases}
	where $\mA, \mB_u,$ and $\mD_u$ are the state-space matrices given in Appendix~\ref{app:params}, 	and the functions $\m f(\cdot)$ and $\m h(\cdot)$ are given as
	\begin{subequations}{\label{eq:state_space_gen_param}}
		\begin{align}
%		\m f(\m x, \m u) &= \bmat{f_1(\m x, \m u)\\f_2(\m x, \m u)\\f_3(\m x, \m u)\\f_4(\m x, \m u)},\;\;
%		\m h(\m x, \m u) = \bmat{h_1(\m x, \m u)\\h_2(\m x, \m u)}, \notag \\
		\hspace{-0.0cm}f_1(\m x, \m u) &= -\alpha_1\label{eq:f_1}\\
		\hspace{-0.0cm}f_2(\m x, \m u) &= \alpha_3x_4u_4\cos x_1 - \alpha_3x_3u_4\sin x_1-\alpha_3x_4u_3\sin x_1 \nonumber\\ &\,\quad- \alpha_3x_3u_3\cos x_1 +\alpha_4u_3u_4\cos 2x_1 \nonumber\\ &\,\quad+ \tfrac{1}{2}\alpha_4\left(u_4^2-u_3^2\right)\sin 2x_1 + \alpha_6 \label{eq:f_2}\\
		\hspace{-0.0cm}f_3(\m x, \m u) &= \alpha_8u_4\cos x_1 - \alpha_8u_3\sin x_1\label{eq:f_3}\\
		\hspace{-0.0cm}f_4(\m x, \m u) &= \alpha_{10}u_3\cos x_1 + \alpha_{10}u_4\sin x_1\label{eq:f_4} \\
		\hspace{-0.0cm}h_1(\m x, \m u) &= x_3\cos x_1 + x_4\sin x_1 +\beta_1u_3\sin 2x_1 \nonumber \\ 
		&\quad + \beta_1u_4\cos 2x_1\label{eq:h_1}\\
		\hspace{-0.0cm}h_2(\m x, \m u) &= x_3\sin x_1 - x_4\cos x_1 -\beta_1u_3\cos 2x_1 \nonumber \\ 
		&\quad - \beta_1u_4\sin 2x_1,\label{eq:h_2} 
		\end{align}
	\end{subequations}
	where constants $\alpha_{1,2,\cdot,10}$ and $\beta_{1,2}$ are described in Appendix~\ref{app:params}. The next section provides analytical methods to compute the Lipschitz constants for $\m f(\cdot)$ and $\m h(\cdot)$. 
	
	\section{The Computation of Lipschitz Constant}~\label{sec:class}
	It is evident from \eqref{eq:state_space_gen} that the generator dynamic model is highly nonlinear. As mentioned earlier, it is important to understand the behavior of the nonlinearities involved in $\m f(\cdot)$ and $\m h(\cdot)$.
	%, the nonlinear model \eqref{eq:state_space_gen} can be simply linearized around a certain operating point to obtain a linear dynamic model of the generator. This linear dynamic model is then utilized to perform dynamic state estimation by implementing various methods in state estimation such as linear observers and Kalman filters. Although this leads to a relatively simple implementation, this method is not practical for systems with varying operating points. Alternatively, several types of observer and modified Kalman filter that can take into account generator's nonlinear model are considered. This in turn allows the dynamic state estimation to be performed on generators with varying operating points. Therefore, it is crucial to study the nonlinearity of the generator's model, since characterizing the nonlinearity enables us to design an appropriate observer. 
	By assuming that the state vector $\m x$ and input vector $\m u$ belong to certain compact sets, as stated in Assumption \ref{asmp:bounded_state_and_input}, the characteristics of the nonlinearity can then be studied---either analytically or numerically.  
	%\textcolor{red}{Do we need $\mathbfcal{X}$ and $\mathbfcal{U}$? Since we are only studying single generator models, I suggest we use $\mathbfcal{X}$ and $\mathbfcal{U}$. Please fix this in the entire paper if you agree with me. }
	\begin{myasm}\label{asmp:bounded_state_and_input}
		The state vector $\m x$	and input vector $\m u$ in \eqref{eq:state_space_gen} are bounded such that $\m x \in \mathbfcal{X}$ and $\m x \in \mathbfcal{U}$ where
		\begin{subequations}
			\begin{align}
			\mathbfcal{X} &:= \left[\barbelow{x}_1,\bar{x}_1\right]\times\left[\barbelow{x}_2,\bar{x}_2\right]\times\left[\barbelow{x}_3,\bar{x}_3\right]\times\left[\barbelow{x}_4,\bar{x}_4\right] \\
			\mathbfcal{U} &:= \left[\barbelow{u}_1,\bar{u}_1\right]\times\left[\barbelow{u}_2,\bar{u}_2\right]\times\left[\barbelow{u}_3,\bar{u}_3\right]\times\left[\barbelow{u}_4,\bar{u}_4\right].
			\end{align}
		\end{subequations}
	\end{myasm}
%	That is, for any state $x_i$ and input $u_j$, we have $\barbelow{x}_i\leq x_i\leq \bar{x}_i$ and $\barbelow{u}_j\leq u_j\leq \bar{u}_j$. 
	Realize that the above assumption is practical and holds for most power systems models as physical quantities such as angles and frequencies are naturally bounded. 
%	Actual bounds of states and control inputs are given in Section~\ref{sec:simulations}.  
	Since $\m f(\cdot)$ and $\m h(\cdot)$ are bounded and continuously differentiable,  both are locally Lispchitz continuous.
%	Based on the above assumption, it will be shown in the sequel that both $\m f(\cdot)$ and $\m h(\cdot)$ are locally differentiable Lispchitz continuous functions. To proceed, we start by introducing the definition of Lipschitz continuity. 
The following introduces the definition of Lipschitz continuity. 
	\begin{mydef}\label{def:Lipschitz_cont}
		Let $\m g : \mathbb{R}^{m} \rightarrow \mathbb{R}^{n}$ be a function. Then, $\m g$ is Lipschitz continuous in $\mathbfcal{B} \subseteq \mathbb{R}^{m}$ if there exists a constant $\gamma \geq 0$ such that for all $\m x, \hat{\m x} \in \mathbfcal{B}$
		\begin{align}
		\norm{\m g(\m x)-\m g(\hat{\m x})}_2 \leq \gamma \norm{\m x -  \hat{\m x} }_2. \label{lipschitz}
		\end{align}
	\end{mydef} 
	\noindent The best (ideal) Lipschitz constant is the smallest $\gamma$ satisfying \eqref{lipschitz}. Although desirable, finding such $\gamma$ can be challenging. 
%	However, in most cases, having any $\gamma$ satisfying \eqref{lipschitz} for any vector in the state-space is fairly acceptable. 
	To that end, we use the following lemma to compute a Lipschitz constant $\gamma$ which, albeit not giving the smallest constant, can still be useful for our purpose---see Section \ref{sec:simulations}.
	\begin{mylem}\label{lem:Lemma1}
		Let $\m g :  \mathbb{R}^{m} \rightarrow \mathbb{R}^{n}$ and $\m x, \hat{\m x} \in \mathbfcal{B}$ where $\mathbfcal{B} \subseteq \mathbb{R}^{m}$. If there exists $\gamma_i \geq 0$ such that 
		\begin{align}
		\abs{g_i(\m x)- g_i(\hat{\m x})}  \leq \gamma_i \norm{\m x - \hat{\m x} }_2, \label{lipschitz2}
		\end{align} 
		for all $i = 1,\hdots,n$, then  $\m g $ is Lipschitz continuous in $\mathbfcal{B}$ with Lipschitz constant $\gamma = \sqrt{\sum_{i=1}^{n}\gamma^2_i}$.
	\end{mylem}
%	\begin{proof}
%		Let $\m g : \mathbb{R}^{m} \rightarrow \mathbb{R}^{n}$ and $\m x, \hat{\m x} \in \mathbfcal{B}$ where $\mathbfcal{B} \subseteq \mathbb{R}^{m}$. Then, suppose that for all $g_i : \mathbb{R}^{m} \rightarrow \mathbb{R}$ there exists $\gamma_i\geq0$ such that $\abs{g_i(\m x)- g_i(\hat{\m x})}\leq \gamma_i \norm{\m x - \hat{\m x} }_2$ for all $i = 1,\hdots,n$. Then
%		\begin{align*}
%		\norm{\m g(\m x)-\m g(\hat{\m x})}_2^2 
%		&= \sum_{i=1}^{n}\abs{g_i(\m x)- g_i(\hat{\m x})}^2 \leq \sum_{i=1}^{n}\gamma_i^2 \norm{\m x - \hat{\m x} }_2^2.
%		\end{align*}
%		This shows that $\m g $ is Lipschitz continuous in $\mathbfcal{B}$ with Lipschitz constant $\gamma = \sqrt{\sum_{i=1}^{n}\gamma^2_i}$.
%	\end{proof}
%	The proof of Lemma~\ref{lem:Lemma1} is omitted here for brevity. By virtue of this lemma, it is shown below that $\m f(\cdot)$ and $\m h(\cdot)$ for the generator dynamics and PMU measurement model in~\eqref{eq:state_space_gen} are in fact locally Lipschitz continuous. %Moreover, since both functions are also continuously differentiable, we show that they have bounded partial derivatives with respect to $\m x$. The next two theorems provide these results.
	The proof of Lemma~\ref{lem:Lemma1} is omitted here for brevity. By virtue of this lemma, the following result presents analytical formulations to compute the corresponding Lipschitz constants for the two functions.
	\begin{theorem}\label{thm:lipschitz_constant}
		Consider $\m f:\mathbb{R}^4\times\mathbb{R}^4\rightarrow \mathbb{R}^4$ and $\m h:\mathbb{R}^4\times\mathbb{R}^4\rightarrow \mathbb{R}^2$ from \eqref{eq:state_space_gen}. Then, for any $\m x, \hat{\m x}\in \mathbfcal{X}$ and $\m u \in \mathbfcal{U}$, we have
		\begin{subequations}\label{eq:lipschitz_constant_equations}
			\begin{align}
			\norm{\m f(\m x, \m u)-\m f(\hat{\m x}, \m u)}_2 \leq \gamma_f \norm{\m x -  \hat{\m x}}_2 \label{eq:lipschitz_constant_equation_f}\\
			\norm{\m h(\m x, \m u)-\m h(\hat{\m x}, \m u)}_2 \leq \gamma_h \norm{\m x -  \hat{\m x}}_2,\label{eq:lipschitz_constant_equation_h}
			\end{align}
		\end{subequations}
		where constants $\gamma_f$ and $\gamma_h$ are given as
		\begin{subequations}\label{eq:lipschitz_constants}
			\begin{align}
			\gamma_f &=\sqrt{ \tilde{\gamma}_{f}^2+\left(\abs{\alpha_8}^2+\abs{\alpha_{10}}^2\right)\left(\kappa_{\mathrm{u}3}+\kappa_{\mathrm{u}4}\right)^2}\label{eq:lipschitz_constants_f}\\
			\gamma_h &= \sqrt{2}\left(\kappa_{\mathrm{x}3}+\kappa_{\mathrm{x}4}+2\abs{\beta_1}\left(\kappa_{\mathrm{u}3}+\kappa_{\mathrm{u}4}\right)+\sqrt{2}\right),\label{eq:lipschitz_constants_h}
			\end{align}
			and $\tilde{\gamma}_{f}$ in \eqref{eq:lipschitz_constants_f} is specified as
			\begin{align}
			\tilde{\gamma}_{f} &= \abs{\alpha_3}\left(\left(\kappa_{\mathrm{u}3}+\kappa_{\mathrm{u}4}\right)\left(1+\kappa_{\mathrm{x}3}+\kappa_{\mathrm{x}4}\right)+2\kappa_{\mathrm{u}3}\kappa_{\mathrm{u}4}\right) \nonumber\\
			&\quad + \abs{\alpha_4}\left(\kappa_{\mathrm{u}3}\left(1+\kappa_{\mathrm{u}3}\right)+\kappa_{\mathrm{u}4}\left(1+\kappa_{\mathrm{u}4}\right)\right),\label{eq:lipschitz_constants_f_aux}
			\end{align}
		\end{subequations}
		with $\kappa_{\mathrm{x}i} := \mathrm{max}\left(\abs{\barbelow{x}_i},\abs{\bar{x}_i}\right)$ and $\kappa_{\mathrm{u}j} := \mathrm{max}\left(\abs{\barbelow{u}_j},\abs{\bar{u}_j}\right)$.
	\end{theorem} 
	\begin{proof}
		Let $\m x, \hat{\m x}\in \mathbfcal{X}$ and $\m u \in \mathbfcal{U}$. First, for $f_1(\cdot)$ given in \eqref{eq:f_1}, we have
		\begin{subequations}\label{eq:lipschitz_constants_f_der}
			\begin{align}
			\abs{f_1\left(\m x, \m u\right)-f_1\left(\hat{\m x}, \m u\right)} = \abs{-\alpha_1-(-\alpha_1)} = 0.
			\end{align}
			Next, for $f_2(\cdot)$ given in \eqref{eq:f_2}, we obtain
			\begin{align}
			\vert &f_2\left(\m x, \m u\right)-f_2\left(\hat{\m x}, \m u\right)\vert \leq \abs{\alpha_3 u_4}\left(\abs{x_3\sin x_1-\hat{x}_3\sin \hat{x}_1}\right.\nonumber\\
			& \;\left. + \abs{x_4\cos x_1-\hat{x}_4\cos \hat{x}_1}\right) + \abs{\alpha_3 u_3}\left(\abs{x_3\cos x_1-\hat{x}_3\cos \hat{x}_1}\right.\nonumber\\
			& \;\left. + \abs{x_4\sin x_1-\hat{x}_4\sin \hat{x}_1}\right)+\abs{\alpha_3u_3u_4}\abs{\cos 2x_1-\cos 2\hat{x}_1} \nonumber \\
			&\;+\tfrac{1}{2}\abs{\alpha_4\left(u_4^2-u_3^2\right)}\abs{\sin 2x_1-\sin 2\hat{x}_1}.\nonumber
			\end{align}
			Since $\abs{\sin 2x_1-\sin 2\hat{x}_1}\leq 2\abs{x_1-\hat{x}_1}$, $\abs{\cos 2x_1-\cos 2\hat{x}_1}\leq 2\abs{x_1-\hat{x}_1}$, $\abs{x_i\sin x_1-\hat{x}_i\sin \hat{x}_1}\leq \kappa_{x,i}\abs{x_1-\hat{x}_1}+\abs{x_i-\hat{x}_i}$, and $\abs{x_i\cos x_1-\hat{x}_i\cos \hat{x}_1}\leq \kappa_{x,i}\abs{x_1-\hat{x}_1}+\abs{x_i-\hat{x}_i}$ for $i=3,4$, then we ultimately get
			\begin{align}
			\vert f_2\left(\m x, \m u\right)&-f_2\left(\hat{\m x}, \m u\right)\vert \leq \tilde{\gamma}_{f}\norm{\m x -  \hat{\m x}}_2,
			\end{align}
			where $\tilde{\gamma}_{f}$ is given in \eqref{eq:lipschitz_constants_f_aux}. For $f_3(\cdot)$ given in \eqref{eq:f_3}, we have
			\begin{align}
			\vert f_3\left(\m x, \m u\right)-f_3\left(\hat{\m x}, \m u\right)\vert &\leq \abs{\alpha_8}\abs{u_3}\abs{\sin x_1-\sin \hat{x}_1} \nonumber \\&\quad +\abs{\alpha_8}\abs{u_4}\abs{\cos x_1-\cos \hat{x}_1} \nonumber \\
			&\leq \abs{\alpha_8}\left(\kappa_{\mathrm{u}3}+\kappa_{\mathrm{u}4}\right)\norm{\m x -  \hat{\m x}}_2.
			\end{align}
			Last, for $f_4(\cdot)$ given in \eqref{eq:f_4}, we obtain
			{\small\begin{align}
				\vert f_4\left(\m x, \m u\right)-f_4\left(\hat{\m x}, \m u\right)\vert &\leq \abs{\alpha_{10}}\abs{u_4}\abs{\sin x_1-\sin \hat{x}_1} \nonumber \\&\quad +\abs{\alpha_{10}}\abs{u_3}\abs{\cos x_1-\cos \hat{x}_1} \nonumber \\
				&\leq \abs{\alpha_{10}}\left(\kappa_{\mathrm{u}3}+\kappa_{\mathrm{u}4}\right)\norm{\m x -  \hat{\m x}}_2.
				\end{align}}
		\end{subequations}
		Applying Lemma \ref{lem:Lemma1} to equations \eqref{eq:lipschitz_constants_f_der} yields \eqref{eq:lipschitz_constant_equation_f} with $\gamma_f$ is equals to \eqref{eq:lipschitz_constants_f}. Likewise, for for $h_1(\cdot)$ given in \eqref{eq:h_1}, we have
		\begin{subequations}\label{eq:lipschitz_constants_h_der}
			\vspace{-0.5cm}
			\begin{align}
			&\vert h_1\left(\m x, \m u\right)-h_1\left(\hat{\m x}, \m u\right)\vert \leq \abs{x_4\sin x_1-\hat{x}_4\sin \hat{x}_1}\nonumber \\ &\quad +\abs{x_3\cos x_1-\hat{x}_3\cos \hat{x}_1} +\abs{\beta_1}\abs{u_3}\abs{\sin 2x_1-\sin 2\hat{x}_1} \nonumber \\ &\quad +\abs{\beta_1}\abs{u_4}\abs{\cos 2x_1-\cos 2\hat{x}_1} \nonumber \\
			&\leq \left(\vphantom{\sqrt{2}}\kappa_{\mathrm{x}3}+\kappa_{\mathrm{x}4}+2\abs{\beta_1}\left(\kappa_{\mathrm{u}3}+\kappa_{\mathrm{u}4}\right)\right.\left. + \sqrt{2}\right)\norm{\m x -  \hat{\m x}}_2.
			\end{align}
			Finally, for $h_2(\cdot)$ given in \eqref{eq:h_2}, we obtain
			\begin{align}
			&\vert h_2\left(\m x, \m u\right)-h_2\left(\hat{\m x}, \m u\right)\vert\leq \abs{x_3\sin x_1-\hat{x}_3\sin \hat{x}_1}\nonumber \\ &\quad +\abs{x_4\cos x_1-\hat{x}_4\cos \hat{x}_1}+\abs{\beta_1}\abs{u_3}\abs{\cos 2x_1-\cos 2\hat{x}_1} \nonumber \\ &\quad +\abs{\beta_1}\abs{u_4}\abs{\sin 2x_1-\sin 2\hat{x}_1} \nonumber \\
			&\leq \left(\vphantom{\sqrt{2}}\kappa_{\mathrm{x}3}+\kappa_{\mathrm{x}4}+2\abs{\beta_1}\left(\kappa_{\mathrm{u}3}+\kappa_{\mathrm{u}4}\right)\right.\left. + \sqrt{2}\right)\norm{\m x -  \hat{\m x}}_2.
			\end{align}
		\end{subequations}
		Applying Lemma \ref{lem:Lemma1} to equations \eqref{eq:lipschitz_constants_h_der} yields \eqref{eq:lipschitz_constant_equation_h} with $\gamma_h$ is equals to \eqref{eq:lipschitz_constants_h}, thus completing the proof.
	\end{proof}
	
	%	\subsection{Significance of Theorem~\ref{thm:lipschitz_constant}}
%	Theorem \ref{thm:lipschitz_constant} guarantees that, given the bounds on $\m x$ and $\m u$, the functions $\m f(\cdot)$ and $\m h(\cdot)$ in~\eqref{eq:state_space_gen} are locally Lipschitz with respect to $\m x$. This consequently allows us to use observers for Lipschitz nonlinear systems to perform DSE. 
	%	Specifically, since Lipschitz-based nonlinear observer design problems can be written as necessary and sufficient conditions, then one can guarantee the existence of asymptotic observer 
That is, given the operational range of $\m x$ and $\m u$, the corresponding Lipschitz constants $\gamma_f$ and $\gamma_h$ for $\m f (\cdot)$ and $\m h (\cdot)$ can be computed.
%, as given in \eqref{eq:lipschitz_constants} . 
%The numerical tests given in Section~\ref{sec:simulations} provide some examples. 
Notice that these constants are dependent on $\mathbfcal{X}$ and $\mathbfcal{U}$. For power systems having a large operational range, the resulting constants can be conservative, which is undesirable due to limitations on most Lipschitz-based observers that are only suitable for nonlinear systems with small Lipschitz constants \cite{abbaszadeh2010}. With that in mind, the numerical tests investigate this presumed conservatism. To overcome this \textit{potential} limitation, in the next section we also propose a simple numerical algorithm to {approximate} Lipschitz constants, thereby yielding smaller values of $\gamma_{f}$ and $\gamma_{h}$. 

	\section{Numerical Algorithms to Compute $\gamma_f$ and $\gamma_h$}~\label{sec:numerical-algorithms}
%	For higher order generator models or interconnected power networks consisting of many generators, the structure of the nonlinearity can be more complex than the single generator with PMU measurement model~\eqref{eq:state_space_gen}. This makes it more difficult to compute Lipschitz constants
	%	and bounds on partial derivatives analytically
%	---as we did in Section~\ref{sec:class}. Here, and in addition to the potential conservatism of analytical Lipschitz constants, 
	Here we propose numerical algorithms to approximate Lipschitz constant $\gamma_f$ and $\gamma_h$. 
%	for generator dynamics given in \eqref{eq:state_space_gen}. 
	%Realize that the applicability of these algorithms is not specific for single generator model, as they can also be used f multi-machine power networks model.
	%\subsection{Low Discrepancy Sequences}
	The algorithm presented here essentially works by evaluating sample points randomly generated in the domain of interest.
%	, $\mathbfcal{X}$ and $\mathbfcal{U}$ defined in Assumption~\ref{asmp:bounded_state_and_input}. 
	This technique is usually referred to as a \textit{Monte Carlo method}\cite{dalal2008}. While pure Monte Carlo methods use random sampling technique, \textit{Quasi-Monte Carlo} methods use a pseudo-random technique that utilize \textit{low-discrepancy sequences} (LDS). LDS are essentially sequence of points that are distributed almost equally in the domain. The concept of discrepancy itself can be explained as follows. 
	
%	\footnote{The $n$-dimensional Lebesgue measure is a mathematical description that quantifies the \textit{size} of a set defined in a $n$-dimensional vector space. For example, volume can be regarded as a Lebesgue measure of a cube in $\mathbb{R}^3$.}
	
	Let $\mathbfcal{Z}\subset\mathbb{R}^n$ be the domain of interest and suppose that there are $s$ number of points in that domain so that they can be written as a sequence of points $ \mathcal{S}(\m z, s) := \left\{ \m z_i \right\}^s_{i = 1}$ for each $\m z_i \in \mathbfcal{Z}$. Then, define an interval $\m J \subset \mathbfcal{Z}$ where $\m J := \prod^n_{j=1}[\barbelow{z}_j,\bar{z}_j)$ such that $\barbelow{z}_j\leq z_j < \bar{z}_j$ for all $j = 1,\hdots,n $. That is, $\m J$ defines a $n$-dimensional hypercube in $\mathbfcal{Z}$ specified by lower and upper bounds of each component for each point $\m z_i$ in the sequence $ \mathcal{S}(\m z, s)$. Consider that $\mathcal{P}(\m J)$ denotes the number of points lying in $\m J$ and $\mathcal{V}(\m J)$ denotes the volume  (or $n$-dimensional Lebesgue measure) of $\m J$, then discrepancy $\mathit{D}(\cdot)$ is a measure formally defined as \cite{dalal2008} 
	\begin{align*}
	\mathit{D}(\m J, \mathcal{S}) := \abs{\frac{\mathcal{P}(\m J)}{s}-\mathcal{V}(\m J)},
	\end{align*}
	The quantity $\mathit{D}(\cdot)$ quantifies the difference between the density of $\m J$ (the proportion of points in $\m J$ compared to the all points in the sequence) and the volume of $\m J$ (the proportion of the size of $\m J$ compared to the size of $\mathbfcal{Z}$). If there is a collection of $m$ intervals called $\mathbfcal{J}$ such that $\m J_k \in \mathbfcal{J}$ for $1\leq k \leq m$, then the \textit{star-discrepancy} can be regarded as the worst-case discrepancy \cite{dalal2008}, i.e., $	\mathit{D}^*(\mathcal{S}) := \sup_{\m J\in \mathbfcal{J}}\mathit{D}(\m J, \mathcal{S})$.
%	\begin{align*}
%
%	\end{align*}
	If $\mathcal{S}(\m z, s)$ is a LDS, then 
	%\begin{align*}
	$\lim_{s\to \infty} \mathit{D}^*(\mathcal{S}) = 0,$
	%\end{align*}
	i.e., the worst-case discrepancy is getting smaller as the number of sample points increases \cite{chakrabarty2017}.
	%\color{red}
	%Comments:
	%\begin{itemize}
	%	\item What is the intuition behind $\m J := \prod^n_{j=1}[\barbelow{z}_j,\bar{z}_j)$? You do not explain that properly. Please define better this function. And why does it have `)' not `]' after $\bar{z}_j$?
	%	\item You should write couple sentences on the idea behind  
	%	$$ \lim_{s\to \infty} \sup_{\mathcal{S}(\m z, s)\in \m J} \mathit{D}(\m J, s) = 0.$$
	%	I see your point, and people should understand it, but you immediately jump into LDS examples in the next paragraphs. Please explain this equation more.  
	%	\item Lebesgue measure is analogous in low-dimensions to  space and volume. You mention $n$-dimensional Lebesgue measure but a lot of power systems people do not what that means. So: add a footnote maybe to define what this means? Give a basic intuition to what it means please.
	%\end{itemize}
	%
	%\normalcolor
LDS typically produce more accurate results than random sampling techniques in numerical Monte Carlo integration, as discussed in \cite{dalal2008}. 
%	Some numerical-based algorithms that have benefited from LDS are {explicit model predictive control} \cite{chakrabarty2017}, genetic algorithm \cite{kimura2006}, multi-level single linkage algorithm for global optimization \cite{kucherenko2005}, and economical valuation of complex securities \cite{tan2000}. 
	\textit{Halton}, \textit{Halton Leaped}, \textit{Sobol}, and \textit{Niederreiter} sequences
	%~\cite{kocis1997computational,cheng2000computational,dalal2008} 
	are examples of well known LDS.
	Further details explaining methods to generate these sequences can be found in \cite{cheng2000computational}.

	% \subsection{Numerical Lipschitz Constants Computation}
	Given the above discussion, we provide an algorithm to approximate Lipschitz constants. From the knowledge of state and input bounds from Assumption \ref{asmp:bounded_state_and_input}, we first generate $s$ number of sample points in $\mathbfcal{X}$ and $ \mathbfcal{U}$ using the aforementioned LDS.
	After generating such points, one can approximate the Lipschitz constants $\gamma_f$ and $\gamma_h$ by using the definition of Lipschitz constant given in \eqref{lipschitz}. %\textcolor{red}{The next two sentences are confusing. You want to say that there's exponential number of iterations to perform if we wanna use \eqref{lipschitz} for a fixed number of samples $s$ (can you count how many combinations, btw?). But you made it more difficult to understand this. Please fix it.} 
	If this method is pursued, then ideally the algorithm has to run $N$ iterations where $N$ is the number of combination of all sample points, which is equal to $N = \comb{s}{2} = \frac{s!}{2!(s-2)!}$. For a large number of sample points $s$, this method requires a huge number of iterations $N$ which consequently increases the computational burden. As an alternative, knowing that the $\m f(\cdot)$ and $\m h(\cdot)$ are continuously differentiable, the {numerical} Lipschitz constants $\gamma_f$ and $\gamma_h$ can be computed by taking the supremum of the norm of Jacobian \cite{abbaszadeh2010} of the nonlinear functions $\m f(\cdot)$ and $\m h(\cdot)$, that is $$\gamma_f= \sup_{\m x \in \mathbfcal{X}, \m u \in \mathbfcal{U}} \norm{\mathrm{D}_{\mathrm{\m x}}\m f}_2,\;\;\;\;\;\gamma_h=\sup_{\m x \in \mathbfcal{X}, \m u \in \mathbfcal{U}} \norm{\mathrm{D}_{\mathrm{\m x}}\m h}_2.$$ 
	%	as discussed in~\cite{dekker1984stability,Scaman2018}. 
	
	%	Similarly for $\m h(\cdot)$, the constant $\gamma_h$ can be computed as 
	
	\begin{algorithm}[t]
		\caption{Numerical Computation of $\gamma_{f}$ and $\gamma_{h}$}\label{algoLC}
		\DontPrintSemicolon
		%			\begin{algorithmic}
		\textbf{input:} $\m f(\cdot)$, $\m h(\cdot)$, $\mathbfcal{X}$,  $\mathbfcal{U}$, $s$\;
		\textbf{generate:} $s$ sample points in $\mathbfcal{X}$ and $\mathbfcal{U}$ \;
		\textbf{initialize:}  $\gamma_f \leftarrow  -\infty $, $\gamma_h \leftarrow  -\infty$\;
		\For{$i=1:s$}{
			$\m x \leftarrow \m x_i \in \mathbfcal{X}$, $\m u \leftarrow \m u_i \in \mathbfcal{U}$\;
			%		$\m F \leftarrow \dfrac{\partial \m f(\m x)}{\partial \m x}$, \,
			%		$\m H \leftarrow \dfrac{\partial \m h(\m x)}{\partial \m x}$\;
			$\gamma_{f_{i}} \leftarrow    \norm{\mathrm{D}_{\mathrm{\m x}}\m f}_2$, $\gamma_{h_{i}} \leftarrow    \norm{\mathrm{D}_{\mathrm{\m x}}\m h}_2$\;
			$\gamma_f \leftarrow \max(\gamma_{f_{i-1}}, \gamma_{f_{i}})$,   
			$\gamma_h \leftarrow \max(\gamma_{h_{i-1}}, \gamma_{h_{i}})$\;} 
		\textbf{output:} $\gamma_f^{\mathrm{(numerical)}}$ and $\gamma_h^{\mathrm{(numerical)}}$
		%			\end{algorithmic}
	\end{algorithm}
%		\vspace*{-0.5cm}
	
	Algorithm~\ref{algoLC} illustrates an offline search method to obtain $\gamma_h$ and $\gamma_f$. Realize that this algorithm only repeats $s$ times, which is exactly equal to the number of sample points. 
	
		\section{A Lipschitz-Based Observer for DSE}\label{sec:observer}
	We now explore how these Lipschitz constants can be utilized to perform DSE by implementing a Lipschitz-based observer from~\cite{phanomchoeng2010}. Since this particular observer does not consider nonlinear output measurement model, we simply use a linearized measurement model to synthesize the observer gain matrix. The observer's dynamics are constructed as
	\begin{subequations}\label{eq:nonlinear_observer_dynamics}
		\begin{align}
		\dot{\hat{\m x}} &= \mA \hat{\m x}+\boldsymbol{f}(\hat{\m x},\boldsymbol{u}) + \m{B_\mathrm{u}} \m u + \mL(\m y-\hat{\m y})\\
		\hat{\m y} &= \mC \hat{\m x} + \m{D_\mathrm{u}} \m u,
		%,
		\end{align}
	\end{subequations}
	where $\mL$ is the observer gain matrix and $\mC$ and $\m{D_\mathrm{u}}$ are obtained by linearizing $\m h(\cdot)$ around a certain operating point. 
%	initial conditions $\m x_0$ and $\m u_0$. 
	To obtain $\mL$, the following linear matrix inequality (LMI) is then solved
	\begin{align}
	\bmat{\mA^{\top}\mP+\mP\mA-\mC^{\top}\mY^{\top}-\mY\mC+\eta\gamma^2_f\mI & \mP^{\top} \\ \mP & -\eta\mI}\prec 0, \label{eq:lmi_rajamani}
	\end{align}
	where the variables are $\mP =\m P^{\top} \succ 0$, $\mY$, and $\eta \geq 0$; $\gamma_f$ denotes the corresponding (analytical or numerical) Lipschitz constant for $\m f(\cdot)$. After solving the LMI, the observer gain matrix can be simply computed as $\mL = \mP^{-1}\mY$. In the next section, we compare the analytical and numerical Lipschitz constants and utilize them for performing DSE using the aforementioned Lipschitz-based observer.

		\begin{figure*}[t]
		\centering
		\includegraphics[page=1,scale=0.94]{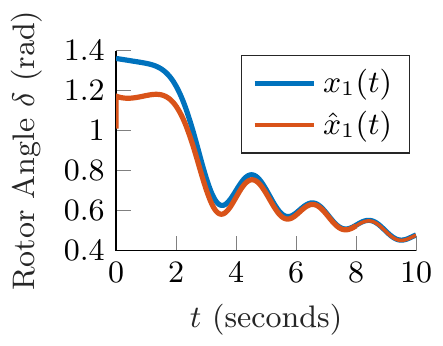}		\includegraphics[page=2,scale=0.94]{standalone_fig}		
		\includegraphics[page=3,scale=0.94]{standalone_fig}							\includegraphics[page=4,scale=0.94]{standalone_fig}	
		%		\vspace{-0.75cm}
		\caption{System's and observer's trajectories considering $\gamma_f = \gamma_f^{(\mathrm{analytical})}$ for the four generator states ($\delta, \omega, e_q', e_d'$). Similar state estimation results are obtained when using numerical Lipschitz constants after obtaining the corresponding observer gain $\m L$.}
		\label{Fig:analytical}	
%		\vspace{-0.3cm}
	\end{figure*}
	
	\section{Numerical Simulations}~\label{sec:simulations}
	This section investigates the property and characteristic of the proposed analytical and numerical methods to determine the corresponding Lipschitz constants for $\m f(\cdot)$ and $\m h(\cdot)$. First, we compare the values of Lipschitz constants obtained from using both methods and second, explore the impact of the \textit{potentially} conservative analytical Lipschitz constants on the design of asymptotic observers for the nonlinear generator with PMU measurement models~\eqref{eq:state_space_gen}, with the objective of performing DSE.
%	This section investigates the conservatism in the derived Lipschitz constants
%	%	and bounded Jacobian analytical constants 
%	from Section~\ref{sec:class}, when compared with the computational methods from Section~\ref{sec:numerical-algorithms}. To that end, we study the aforementioned problem in this section and explore the impact of \textit{potentially} too conservative analytical Lipschitz constants on the design of asymptotic observers for the nonlinear generator with PMU measurement models~\eqref{eq:state_space_gen}, with the objective of performing DSE.
	
%	While it has been hitherto hypothesized in the control theoretic literature that analytical Lipschitz constants might be too conservative, there is a need to analytically and numerically investigate this hypothesis for generator dynamics. 
	
	\subsection{Power System Parameters and Setup} We test the proposed approaches on the 16-machine, 68-bus system that is extracted from the PST toolbox \cite{chow1992toolbox} by considering Generator 16 in the  network. The parameters are obtained from \cite{chow1992toolbox}.
	%The SMIB system parameters are extracted from the PST toolbox \cite{chow1992toolbox} and is shown in Figure \ref{SMIB}. 
	%In order to generate dynamic response, we apply a three-phase fault at bus $6$ of line $6-11$ and is cleared at the near and remote ends after $0.05$ s and $0.1$ s. 
	%The study is performed for generator 1, %the post-contingency system, 
	%for which the parameters are listed in Appendix~\ref{app:parambounds} and Table \ref{para1}.
	The input vector $\boldsymbol{u}$, including $T_{\mathrm{m}i}$, $E_{\mathrm{fd}i}$, $i_{\mathrm{R}i}$, and $i_{\mathrm{I}i}$ are obtained from simulations of the whole system in which each generator is using a transient model with IEEE Type DC1 excitation system and
	a simplified turbine-governor system \cite{qi2018comparing}. To obtain lower and upper bounds on the state $\boldsymbol{x}$ and input $\boldsymbol{u}$, their minima and maxima are measured. 
%	and displayed in Table \ref{range}. 
	All simulations are conducted by using %MATLAB R2017b running on a 64-bit Windows 10 with 2.5GHz Intel\textsuperscript{R} Core\textsuperscript{TM} i7-6500U CPU and 8 GB of RAM.
	MATLAB R2016b running on a 64-bit Windows 10 with 3.4GHz Intel\textsuperscript{R} Core\textsuperscript{TM} i7-6700 CPU and 16 GB of RAM. 
	We use YALMIP \cite{Lofberg2004} as the interface and MOSEK \cite{andersen2000mosek} solver to get the solutions of the LMI required by the DSE and observer.

	%\begin{figure}[h]
	%\vspace{-2.2cm}	\centering
	%	\includegraphics[scale=0.3,angle=90]{SMIB}
	%	\vspace{-2.2cm}	\centering
	%	\caption{Single machine infinite bus power system.}
	%	\label{SMIB}
	%\end{figure} 

%	\begin{table}[h]
%		\scriptsize	\renewcommand{\arraystretch}{1.3}
%		\caption{State $\m x$ and input $\m u$ operational range $\mathbfcal{X}$ and $\mathbfcal{U}$.}
%		\label{range}
%		\centering
%		\begin{tabular}{ccccc}
%			\hline
%			Parameter & State/Input & Unit & Min & Max\\
%			\hline
%			\hline
%			$\delta$ & $x_1$ & rad & 0.4605 & 1.3607 \\
%			\hline
%			$\omega$ & $x_2$ & rad/s & 376.4 & 377.2 \\
%			\hline
%			$e'_{\mathrm{q}}$ & $x_3$ & p.u. & 0.9454 & 0.3920 \\
%			\hline
%			$e'_{\mathrm{d}}$ & $x_4$ & p.u. & 1.1984 & 0.4785 \\
%			\hline
%			$T_{\mathrm{m}}$ & $u_1$ & p.u. & 0.3631 & 0.3635 \\
%			\hline
%			$E_{\mathrm{fd}}$ & $u_2$ &p.u. & 1.2450 & 1.2702 \\
%			\hline
%			$i_{\mathrm{R}}$ & $u_3$ & p.u. & 28.5280 & 30.1034 \\
%			\hline
%			$i_{\mathrm{I}}$ & $u_4$ & p.u. & 26.6607 & 28.2618 \\ 
%			\hline
%		\end{tabular}
%		\vspace{-0.5cm}
%	\end{table}
	
	\subsection{Lipschitz Constants Computation}
	This section is devoted to determine Lipschitz constants 
	%	and bounds on partial derivatives 
	of $\m f(\cdot)$ and $\m h(\cdot)$ given the generator parameters and operational range. First, we compute the analytical Lipschitz constants $\gamma_f^{(\mathrm{analytical})}$ and $\gamma_h^{(\mathrm{analytical})}$ by using applying \eqref{eq:lipschitz_constants_f} and \eqref{eq:lipschitz_constants_h} from Theorem~\ref{thm:lipschitz_constant}. Second, we implement Algorithm \ref{algoLC} to compute numerical approximations of Lipschitz constants of $\m f(\cdot)$ and $\m h(\cdot)$. For these approximations, we utilize three different methods to generate the sampled points: random, Sobol, and Halton sequences. By generating $2000$ sample points inside the sets $\mathbfcal{X}$ and $\mathbfcal{U}$, we run the algorithm ten times to minimize the effect of randomization. The corresponding MATLAB functions used to generate these points are \texttt{rand}, \texttt{sobolset}, and \texttt{haltonset}. %Considering that these functions generate points with values between 0 and 1, we then scale the generated points such that they lie inside the domain $\mathbfcal{X}$ and $\mathbfcal{U}$. 
	The results are given in Table \ref{result}, where the mean values of the approximated Lipschitz constants $\gamma_{f,h}^{(\mathrm{numerical})}$ are compared with the analytical Lipschitz constants. 
	
	\begin{table}[t]
		%	\vspace{-0.5cm}
			\renewcommand{\arraystretch}{1.3}
		\caption{Analytical versus Numerical Lispchitz Constants.}
		\label{result}
		\centering
		\begin{tabular}{ccccc}
			\hline
			Constant & Analytical & Random & Sobol & Halton\\
			\hline
			\hline
			$\gamma_f$ & 715.395 & 19.802 & 20.128 & 20.131 \\
			\hline
			$\gamma_h$ & 5.390 & 1.631 & 1.629 & 1.630 \\
			\hline
		\end{tabular}
%		\vspace{-0.3cm}
	\end{table}
	
	From this table, we observe that the analytical Lipschitz constants are much higher than the numerical ones, especially for $\gamma_f$. This is the case because the analytical Lipschitz constants given in  
	\eqref{eq:lipschitz_constants} do not necessarily give the best ones, and thus serve as upper bounds for the numerical Lipschitz constants. We found that the high values are also due to the large operational range of the fourth control input ($i_{I}$) which significantly increases $\gamma_f^{\mathrm{(analytical)}}$. 
	Amid this discrepancy, $\gamma_f^{(\mathrm{analytical})}$ and $\gamma_f^{(\mathrm{numerical})}$ are tested in the next section for performing DSE on a single generator.   Specifically, we investigate whether these conservative analytical constants can be useful to perform DSE. 
	We also observe that using LDS here did not have a drastic impact on the computation of the numerical Lipschitz constants---when compared with random sampling inside $\mathbfcal{X}$ and $\mathbfcal{U}$. To the best of our knowledge, one important feature of LDS for this particular purpose is that the approximated Lipschitz constants will converge to the actual ones as the number of sample point increases (assuming that $\m f(\cdot)$ and $\m h(\cdot)$ have continuous partial derivatives), which may or may not be the case for random sampling.
%	This is surprising, given what is known about LDS; see the discussion in Section~\ref{sec:numerical-algorithms}. More thorough numerical tests are hence required to investigate this unintuitive observation.
	%	\footnote{We also performed   
	%We remind the reader again that the main contribution of the paper is proving that a Lipschitz-based observer can be used, regardless 
	
	%We also compute the bounds of the partial derivatives of $\m f(\cdot)$ and $\m h(\cdot)$ by using \eqref{eq:jacobian_bounds_equations}.
	
	\subsection{Generator DSE Using Lipschitz-Based Observer} 
	In this simulation, we compare two different scenarios where the first one uses $\gamma_f^{(\mathrm{analytical})}=715.395$ whereas the other uses $\gamma_f^{(\mathrm{numerical})}=20.131$, which is the result from particularly using Halton sequence from Table \ref{result}. Note that when simulating the DSE method through the observer~\eqref{eq:nonlinear_observer_dynamics}, the nonlinear model of the output equation for both system and observer are used, i.e., $\m y=\m h(\m x, \m u)$ and $\hat{\m y}=\m h(\hat{\m x}, \m u)$.
	
	Fig. \ref{Fig:analytical} depicts the state estimation trajectories in comparison with the system's trajectories given the analytical Lipschitz constant. Note that we have used significantly different initial conditions $\hat{\m x}(0)$ for the observer, in comparison with the generator's actual initial conditions (this can also be seen from Fig.~\ref{Fig:analytical}). The simulation using the numerical Lipschitz constant exhibits very similar results.
	%	The norm of estimation error for the analytical and numerical approaches are depicted in Fig \ref{Fig:error_norm}. 
	%In both scenarios, the generator states can be succesfully estimated with al no significant difference between those two. 
	This implies that---for this specific test at least---both analytical and numerical Lipschitz constants can be utilized for performing DSE via Lipschitz-based nonlinear observers, and while the analytical Lipschitz constant was in fact large, it hinders neither finding a feasible solution for the LMI~\eqref{eq:lmi_rajamani} nor obtaining asymptotically stable estimation error.  
	
	\section{Summary, Closing Remarks, and Future Work}
	Motivated by the need to study higher-order nonlinear, dynamic models of power networks, this paper deals with the problem of determining the Lipschitz constants for fourth order generator dynamics with PMU measurements, which leads to the investigation of different methods to compute Lipschitz constants: analytical formulation and numerical algorithm based on low discrepancy sampling methods. Numerical tests showcase the discrepancy between the analytical and numerical methods, and applications to DSE of generator states given PMU measurements are provided.
	
	We conclude the paper with the following remarks. \textit{(a)} Albeit conservative, Theorem~\ref{thm:lipschitz_constant} and the analytical Lipschitz constants give confidence in applying Lipschitz-based estimators---and potentially state-feedback controllers for the nonlinear power network. \textit{(b)} Although it is worried that large Lipschitz constants can impede the application of Lipschitz-based observers \cite{abbaszadeh2010}, we found that this may not always be the case, at least for  performing DSE on a single generator. \textit{(c)} Using LDS, in comparison with random sequences, does not seem to highly impact the values of numerical Lipschitz constants. The above observations \textit{(a)}--\textit{(c)} are, however, not thoroughly conclusive. Future work will focus on performing extensive numerical tests for various generators and operating conditions, as well as designing robust observers that consider nonlinear PMU measurement model under uncertainty.
	
%	Extensive numerical tests for various generators and operating conditions are therefore needed. Future work will focus on the aforementioned investigations, as well as designing robust observers that consider nonlinear PMU measurement model under uncertainty.

%	, rather than linearizing around operating points.

	%		\begin{figure}[h]
	%			\centering
	%			\vspace{-0.35cm}
	%			\input{e_norm.tex}
	%			\caption{Norm of estimation error ($|| \hat{\m x}(t)-{\m x}(t)||$) for both scenarios.}
	%			\label{Fig:error_norm}	
	%			\vspace{-0.5cm}
	%		\end{figure}  
	
	%After solving the LMI, we obtain
	%\begin{align*}
	% \mL^{(analytical)} = \bmat{ -10311.4& 7413.05\\
	% -5754.85&  7932.0\\
	%  16949.8& 38596.2\\
	%  5517.24&  7535.0} \\ 
	%\mL^{(numerical)} = \bmat{ -10311.4& 7413.05\\
	% -5754.85&  7932.0\\
	%  16949.8& 38596.2\\
	%  5517.24&  7535.0} \\
	%\end{align*}
	
	\section*{Acknowledgments}
	We gratefully acknowledge the constructive comments from the editor and the reviewers. We also acknowledge the financial support from National Science Foundation through Grant CMMI-DCSD-1728629.

	\bibliographystyle{IEEEtran}	\bibliography{bibl}
%	\vspace{-0.48cm}
	
	\appendices
	
	\section{Single Generator State-Space Parameters}~\label{app:params}
	Constants $\alpha_{1,\ldots, 10}$ and $\beta_{1,2}$ are given as: \begin{subequations} \label{para}
		\begin{align*}
		&\alpha_1=\omega_0,\;\alpha_2=\frac{\omega_0}{2H_i},\;\alpha_3=\frac{\omega_0}{2H_i}\left(\frac{S_\textrm{B}}{S_{\textrm{N}i}}\right) \\
		&\alpha_4=\frac{\omega_0}{2H_i}\left(\frac{S_\textrm{B}}{S_{\textrm{N}i}}\right)^2(x'_{\textrm{q}i}-x'_{\textrm{d}i}),\;\alpha_5=\frac{K_{\textrm{D}i}}{2H_i},\;\alpha_6=\frac{K_{\textrm{D}i}}{2H_i}\omega_0\\
		&\alpha_7=\frac{1}{T'_{\textrm{d0}i}},\;\alpha_8=\frac{1}{T'_{\textrm{d0}i}}\left(\frac{S_\textrm{B}}{S_{\textrm{N}i}}\right)(x_{\textrm{d}i}-x'_{\textrm{d}i}) ,\;\alpha_9=\frac{1}{T'_{\textrm{q0}i}} \\
		&\alpha_{10}=\frac{1}{T'_{\textrm{q0}i}}\left(\frac{S_\textrm{B}}{S_{\textrm{N}i}}\right)(x_{\textrm{q}i}-x'_{\textrm{q}i}),\;\beta_1=\frac{1}{2}\left(\frac{S_\textrm{B}}{S_{\textrm{N}i}}\right)(x'_{\textrm{q}i}-x'_{\textrm{d}i}) \\
		&\beta_2=\frac{1}{2}\left(\frac{S_\textrm{B}}{S_{\textrm{N}i}}\right)(x'_{\textrm{q}i}+x'_{\textrm{d}i}).
		\end{align*}
	\end{subequations}
	The state-space matrices $\m A, \m B_u,$ and $\m D_u$ are given as
	\begin{subequations} \label{para}
		\begin{align*}
		\m A = &\bmat{0&1&0&0\\0&-\alpha_5&0&0\\0&0&-\alpha_7&0\\0&0&0&-\alpha_9},\;
		\m{B_\mathrm{u}} = \bmat{0&0&0&0\\\alpha_2&0&0&0\\0&\alpha_7&0&0\\0&0&0&0} \\
		\m{D_\mathrm{u}} = &\bmat{0&0&0&\beta_2\\0&0&-\beta_2&0}.\;
		\end{align*}
	\end{subequations}
	
	%\section{Generator State and Input Bounds}~\label{app:parambounds}
	
	%
	%\begin{table}[h]
	%	\small	\renewcommand{\arraystretch}{1.3}
	%	\caption{Generator $i$ Parameters \cite{chow1992toolbox}.}
	%	\label{para1}
	%	\centering
	%	\begin{tabular}{ccc}
	%		\hline
	%		Parameter & Description & Value \\
	%		\hline
	%		\hline
	%		$\omega_0$ & \tabincell{l}{Rated value of angular \\ frequency (rad/s)} & \tabincell{l}{$2\pi\cdot 60$ \\ $\cong$ 376.9911} \\
	%		\hline
	%		$S_B, S_N$ & \tabincell{l}{System base MVA and \\ generator base MVA} & 100, 11000 \\
	%		\hline
	%		$K_D, H$ & \tabincell{l}{Damping factor and \\ inertia constant} & 4.45, 4.45 \\
	%		\hline
	%		$T'_{q0}, T'_{d0}$ & \tabincell{l}{Open-circuit time constants \\ at the $q$ and $d$ axes}  & 1.5, 7.8   \\
	%		\hline
	%		$x_q, x_d$ & \tabincell{l}{Synchronous reactance \\ at the $q$ and $d$ axes} & 1.6888, 1.8  \\
	%		\hline
	%		$x'_q, x'_d$ & \tabincell{l}{Transient reactance \\ at the $q$ and $d$ axes} & 0.3590, 0.3590   \\
	%		\hline
	%		%$T_m$ & Mechanical torque &  0.0908 \\
	%		%\hline
	%		%$E_{fd}$ & Excitation output voltage & 1.0641 \\
	%		%\hline
	%		%$\bar{y}$ & \tabincell{l}{Line admittance of \\ the reduced network} & -1.906$i$ \\
	%		%\hline
	%		%$e'_{qI}$ &  \tabincell{l}{transient voltage along \\ $q$ axis of the infinite bus} & 1.0810 \\
	%		%\hline
	%	\end{tabular}
	%\end{table}

\end{document}